\documentclass[11pt]{article} 

\usepackage{fullpage}

\usepackage{amsmath,amsthm, amssymb}
\usepackage{latexsym}
\usepackage{graphicx}
\usepackage{ifthen}
\usepackage{subcaption}
\usepackage{multicol}
\usepackage{algorithm,algorithmic}
\usepackage[numbers]{natbib}
\usepackage{mathtools}
\usepackage{dsfont}
\usepackage{nicefrac}

\usepackage{hyperref}
\hypersetup{colorlinks=true,linkcolor=blue,citecolor=blue,urlcolor=blue}

\newtheorem*{theorem*}{Theorem}
\newtheorem*{definition*}{Definition}

\newtheorem{lemma}{Lemma}

\usepackage{thm-restate}

\newcommand\iid{\stackrel{\mathclap{\tiny\mbox{i.i.d.}}}{ \ \sim \ }}
\renewcommand{\algorithmiccomment}[1]{\bgroup\hfill\footnotesize~#1\egroup}

\usepackage{hyperref}
\hypersetup{colorlinks=true,linkcolor=blue,citecolor=blue,urlcolor=blue}

\newcommand{\algone}{{\textsc{Iterative-Filtering}}}
\newcommand{\algtwo}{{\textsc{Amortized-Filtering}}}
\newcommand{\filter}{{\textsc{Filter}}}
\newcommand{\algMainAlmostFull}{{\textsc{Amortized-Filtering-Proxy}}}
\newcommand{\algMainFull}{{\textsc{Amortized-Filtering-Full}}}

\newcommand{\algdown}{{\textsc{Down-Sampling}}}
\newcommand{\algup}{{\textsc{Up-Sampling}}}
\newcommand{\algadaptive}{{\textsc{Adaptive-Sampling}}}

\newcommand{\estOne}{\textsc{EstimateSet}}
\newcommand{\estTwo}{\textsc{EstimateMarginal}}

\newcommand{\INPUT}{\item[{\bf Input:}]}

\DeclareMathOperator{\E}{\mathbb{E}}
\newcommand{\EU}[1][]{\underset{#1}{\E}}

\DeclareMathOperator{\argmax}{argmax}

\DeclareMathOperator{\poly}{poly}
\newcommand{\OPT}{\texttt{OPT}}

\newcommand{\R}{\mathbb{R}}                     
\newcommand{\U}{\mathcal{U}}

\newcommand{\D}{\mathcal{D}}
\renewcommand{\O}{O}

\title{An Exponential Speedup in Parallel Running Time for Submodular Maximization without Loss in Approximation}

\author{Eric Balkanski\\Harvard University \\ ericbalkanski@g.harvard.edu 
\and Aviad Rubinstein\\ Harvard University \\ aviad@seas.harvard.edu
\and Yaron Singer\\ Harvard University \\ yaron@seas.harvard.edu\\}

\date{}

\begin{document}

\setcounter{page}{0}

\maketitle
\begin{abstract}
In this paper we study the adaptivity of submodular maximization.  Adaptivity quantifies the number of sequential rounds that an algorithm makes when function evaluations can be executed in parallel.  Adaptivity is a fundamental concept that is heavily studied across a variety of areas in computer science, largely due to the need for parallelizing computation.  For the canonical problem of maximizing a monotone submodular function under a cardinality constraint, it is well known that a simple greedy algorithm achieves a $1-1/e$ approximation~\cite{NWF78} and that this approximation is optimal for polynomial-time algorithms~\cite{nemhauser1978best}.  Somewhat surprisingly, despite extensive efforts on submodular optimization for large-scale datasets, until very recently there was no known algorithm that achieves a constant factor approximation for this problem whose adaptivity is sublinear in the size of the ground set $n$.
  
Recent work by~\cite{BS18} describes an algorithm that obtains an approximation arbitrarily close to $1/3$ in $\mathcal{O}(\log n)$ adaptive rounds and shows that no algorithm can obtain a constant factor approximation in $\tilde{o}(\log n)$ adaptive rounds.  This approach achieves an exponential speedup in adaptivity (and parallel running time) at the expense of approximation quality.  

In this paper we describe a novel approach that yields an algorithm whose approximation is arbitrarily close to the optimal $1-1/e$ guarantee in $\mathcal{O}(\log n)$ adaptive rounds.  \emph{This algorithm therefore achieves an exponential speedup in parallel running time for submodular maximization at the expense of an arbitrarily small loss in approximation quality}.  This guarantee is optimal in both approximation and adaptivity, up to lower order terms.
\end{abstract}

\newpage

\section{Introduction}
In this paper we study the adaptivity of submodular maximization.  For the canonical problem of maximizing a non-decreasing submodular function under a cardinality constraint it is well known that the celebrated greedy algorithm which iteratively adds elements whose marginal contribution is largest achieves a $1-1/e$ approximation~\cite{NWF78}.  Furthermore, this approximation guarantee is optimal for any algorithm that uses polynomially-many value queries~\cite{nemhauser1978best}.   

The optimal approximation guarantee of the greedy algorithm comes at a price of high \emph{adaptivity}. 
Informally, the \emph{adaptivity} of an algorithm is the number of sequential rounds it makes when polynomially-many function evaluations can be executed in parallel in each round.  
For submodular maximization, any algorithm whose adaptivity is $t$ can also be executed in $\tilde{\mathcal{O}}(t)$ parallel time in standard parallel computation models (see Section~\ref{sec:related}).  For cardinality constraint $k$ and ground set of size $n$, the greedy algorithm is $k$-adaptive since it sequentially adds elements in $k$ rounds.  In each round it makes $\mathcal{O}(n)$ function evaluations to identify and include the element with maximal marginal contribution to the set of elements selected in previous rounds.  In the worst case $k \in \Omega(n)$ and thus the greedy algorithm is $\Omega(n)$-adaptive and its parallel running time is $\Theta(n)$.
  
Since submodular optimization is regularly applied on very large datasets, adaptivity is crucial as algorithms with low adaptivity enable dramatic speedups in parallel computing time.  
Submodular optimization has been studied for well over forty years now, and in the past decade there has been extensive study of submodular maximization for large datasets~\cite{badanidiyuru2012sketching,kumar2015fast, mirzasoleiman2013distributed, BV14, pan2014parallel, BMKK14, MBK15, mirrokni2015randomized, mirzasoleiman2015distributed,barbosa2015power, mirzasoleiman2016fast, barbosa2016new,EMZ17}.  
Somewhat surprisingly however, until very recently, there was no known constant-factor approximation algorithm for submodular maximization whose adaptivity is sublinear in $n$.


In recent work~\cite{BS18} introduce an algorithm for maximizing monotone submodular functions under a cardinality constraint that achieves a constant factor approximation arbitrarily close to $1/3$ in $\mathcal{O}(\log n)$ adaptive steps.  Furthermore, \cite{BS18} show that no algorithm can achieve a constant factor approximation with $\tilde{o}(\log n)$ rounds.  

For constant factor approximations, \cite{BS18} provide an exponential speedup in the parallel run-time for the canonical problem of maximizing submodular functions under a cardinality constraint.  
This exponential improvement in adaptivity comes at the expense of the approximation quality achievable by $\Omega(n)$-adaptive algorithms (e.g. greedy), and raises two fundamental questions: 

\begin{itemize}
\item \emph{Is there an algorithm whose adaptivity is sublinear in the size of the ground set that obtains an approximation arbitrarily close to the optimal $1-1/e$ approximation guarantee?}
\item \emph{Given that a constant factor approximation cannot be obtained in $\tilde{o}(\log n)$ rounds, what is the best approximation achievable in $\mathcal{O}(\log n)$ rounds?}  
\end{itemize}
  
In this paper we address both questions as summarized by our main result:
\begin{theorem*}
For any constant $\epsilon>0$, any non-decreasing submodular function $f:2^{[n]} \to \mathbb{R}$ and $k\in [n]$, there is an algorithm that with probability $1 - o(1)$ obtains a $1-1/e - \epsilon$ approximation to $\max_{S:|S|\leq k}f(S)$  in $\mathcal{O}(\log n)$ adaptive rounds.
\end{theorem*}
The algorithm gives an exponential speedup in parallel running time for maximizing a submodular function.  In particular, our result shows that exponential speedups in parallel computing are possible with arbitrarily small sacrifice in the quality of the approximation achievable in poly-time.

\subsection{Technical overview}


The main goal of this paper is to achieve the optimal $1-1/e$  guarantee in $\mathcal{O}(\log n)$ adaptive steps.  The optimal $1-1/e$ approximation of the greedy algorithm stems from the guarantee that for any given set $S$ there exists an element whose marginal contribution to $S$ is at least a $1/k$ fraction of the remaining optimal value $\OPT - f(S)$.  A standard inductive argument then shows that  iteratively adding the element whose marginal contribution is maximal results in the $1-1/e$ approximation guarantee.  To obtain the $1-1/e$ guarantee in $r=\mathcal{O}(\log n)$ adaptive steps rather than $k$, we could mimic this idea if in each adaptive step we could add a \emph{block} of $k/r$ elements whose marginal contribution to the existing solution $S$ is at least a $1/r$ fraction of $\OPT - f(S)$.   

The entire challenge is in finding  such a block of $k/r$ elements in $\mathcal{O}(1)$ adaptive steps.  A priori, this is a formidable task when $k/r$ is super-constant. In general, the maximal marginal contribution over all sets of size $k/r$ is as low as $(\OPT - f(S))/r$.  Finding a block of size $t$ of maximal marginal contribution in polynomial time is as hard as solving the general problem of submodular maximization under cardinality constraint $t$, which, in general, cannot be approximated within any factor better than $1-1/e$ using polynomially-many queries~\cite{nemhauser1978best}.  Furthermore, we know it is impossible to approximate within any constant approximation in $o(\log n / \log \log n)$ adaptive rounds~\cite{BS18}.

Despite this seeming difficulty, we show one can exploit a fundamental property of submodular functions to identify a block of size $k/r$ whose marginal contribution is arbitrarily close to $(\OPT - f(S))/r$.  In general, we show that for monotone submodular functions, while it is hard to find a set of size $k$ whose value is an arbitrarily good approximation to  $\OPT$, it is actually possible to find a set of size $k / r$ whose value is arbitrarily close to that of $\OPT/r$ in polynomial time for $r = \O(\log n)$, even when $k/r$ is super-constant.
 

In Section~\ref{sec:alg1} we describe an algorithm which progressively adds a subset of size $k/r$ to the existing solution $S$ whose marginal contribution is arbitrarily close to $(\OPT - f(S))/r$. To do so, it uses $\mathcal{O}(\log n)$ rounds in each such progression and it is hence $\mathcal{O}(\log^2 n)$-adaptive.  At a high level, in each iteration that it adds a block of size $k/r$, the algorithm carefully and aggressively filters elements in $\mathcal{O}(\log n)$ rounds by considering their marginal contribution to a random set drawn from a distribution that evolves throughout the filtering iterations.  

In Section~\ref{sec:alg2}, we generalize the algorithm so that, on average, every step of adding a block of $k/r$ elements is done in $\mathcal{O}(1)$ adaptive steps.  The main idea is to consider epochs, which consist of  sequences of iterations such that, in the worst case, an iteration might still consist of $\mathcal{O}(\log n)$ rounds, but the amortized number of rounds per iteration during an epoch is now constant.

\subsection{Paper organization}
We first discuss related work, followed by preliminary definitions and notation below.  In Section~\ref{sec:alg1} we describe and analyze \algone \  which obtains an approximation guarantee arbitrarily close to $1-1/e$ in $\mathcal{O}(\log^2 n)$ adaptive rounds.  In Section~\ref{sec:alg2} we describe and analyze \algtwo \ which obtains the same approximation guarantee in $\mathcal{O}(\log n)$ rounds.

\subsection{Related work}
\label{sec:related}

\paragraph{Parallel computing and depth.}
Adaptivity is closely related to the concept of \emph{depth} in the PRAM model. The \emph{depth} of a PRAM algorithm is the number of parallel steps it takes on a shared memory machine with any number of processors. That is, it is the longest chain of dependencies of the algorithm, including operations which are not queries. There is a long line of study on the design of low-depth algorithms (e.g. \cite{blelloch1996programming, blelloch2011linear,  berger1989efficient, rajagopalan1998primal, blelloch1998fast, blelloch2012parallel}). As discussed in further detail in Appendix~\ref{sec:appdepth}, our positive results extend to the PRAM model and our main algorithm has  $\tilde{\mathcal{O}}(\log^2 n\cdot d_f)$ depth, where $d_f$ is the depth required to evaluate the function on a set. While the PRAM model assumes that the input is loaded in memory, we consider the value query model where the algorithm is given oracle access to a function of potentially exponential size. 

\paragraph{Adaptivity.} The concept of adaptivity is generally well-studied in computer science, largely due to the role it plays in parallel computing, such as in sorting and selection \cite{Val75, Col88, BMW16},  communication complexity \cite{papadimitriou1984communication, duris1984lower, nisan1991rounds},   multi-armed bandits \cite{AAAK17}, sparse recovery \cite{HNC09,IPW11, haupt2009compressive}, and  property testing \cite{CG17, buhrman2012non,chen2017settling}. Beyond being a fundamental concept, adaptivity is important for applications where sequentiality is the main runtime bottleneck. We discuss in detail several such applications of submodular optimization in Appendix~\ref{sec:appapplications}. Somewhat surprisingly, until very recently $\Omega(n)$ was the best known adaptivity required for a constant factor approximation to maximizing a monotone submodular maximization under a cardinality constraint.      
As discussed above, \citep{BS18} give an algorithm that  is $\mathcal{O}(\log n)$-adaptive and achieves an approximation arbitrarily close to $1/3$. They also show that no algorithm can achieve a constant factor approximation with $\tilde{o}(\log n)$ rounds.  The approach and algorithms in this paper are  different than \cite{BS18} and we provide a detailed comparison in Appendix~\ref{sec:appcomparison}.

\paragraph{Map-Reduce.} There is a long line of work on distributed submodular optimization  in the Map-Reduce model \cite{ kumar2015fast, mirzasoleiman2013distributed, mirrokni2015randomized, mirzasoleiman2015distributed,barbosa2015power, barbosa2016new,EMZ17}. Map-Reduce is designed to tackle issues related to massive data sets that  are too large to either  fit or be processed by a single  machine. Instead of addressing distributed challenges, adaptivity addresses the issue of sequentiality, where query-evaluation time is  the main runtime bottleneck and where these evaluations can be parallelized. The existing Map-Reduce algorithms for submodular optimization have adaptivity that is linear in $n$ in the worst-case. This high adaptivity is caused by the  algorithms run on each machine, which are variants of the greedy algorithm and thus have adaptivity at least linear in $k$. Additional discussion about the Map-Reduce model is provided in Appendix~\ref{sec:appmapreduce}.


\subsection{Basic definitions and notation}
\paragraph{Submodularity.}  
For a given function $f:2^{N} \to \mathbb{R}$, the \emph{marginal contribution} of an element $X\subseteq N$ to a set $S \subseteq N$ denoted $f_{S}(X)$ is defined as $f(S \cup X) - f(S)$.  A function $f : 2^N \rightarrow \R$ is \emph{submodular} if for any $S \subseteq T \subseteq N$ and any $a\in N\setminus T$ we have that
$f_S(a) \geq f_T(a).$\footnote{For readability we abuse notation and write $a$ instead of $\{a\}$ when evaluating a singleton.}
A function is \emph{monotone} or \emph{non-decreasing} if $f(S) \leq f(T)$ for all $S \subseteq T$.  A submodular function $f$ is also   \emph{subadditive}, meaning  $f(S \cup T) \leq f(S) + f(T)$ for all $S \subseteq T$.   The size of the ground set is $n = |N|$ and $k$ denotes the cardinality constraint of the given optimization problem $\max_{S:|S|\leq k}f(S)$.

\paragraph{Adaptivity.}  
As standard, we assume access to a \emph{value oracle} of the function s.t. for any $S\subseteq N$ the oracle returns $f(S)$ in $\mathcal{O}(1)$ time.  Given a value oracle for $f$, an algorithm  is $r$-adaptive 
 if every query $f(S)$ for the value of a set $S$  occurs at a round $i \in [r]$ s.t. $S$ is independent of the values $f(S')$ of all other queries at round $i$, with at most $\poly(n)$ queries at every round. In Appendix~\ref{sec:appdepth} we discuss adaptivity and parallel computing.



\section{\algone: An $\mathcal{O}(\log^2 n)$-adaptive Algorithm}
\label{sec:alg1}

In this section, we present the \algone \ algorithm which obtains an approximation arbitrarily close to $1-1/e$ in $\mathcal{O}(\log^2 n)$ adaptive rounds. At a high level, the algorithm iteratively identifies large blocks of elements of high value and adds them to the solution.  There are $\mathcal{O}(\log n)$ such iterations and each iteration requires $\mathcal{O}(\log n)$ adaptive rounds, which amounts to $\mathcal{O}(\log ^2 n)$-adaptivity.  The analysis in this section will later be used as we generalize this algorithm to one that obtains an approximation arbitrarily close to $1-1/e$ in $\mathcal{O}(\log n)$ adaptive rounds.

\subsection{Description of the algorithm}

The \algone \ algorithm consists of $r$ iterations which each add $k/r$ elements to the solution $S$.  To find these elements the algorithm filters out elements from the ground set using the \textsc{Filter} subroutine and then adds a set of size $k/r$ sampled uniformly at random from the remaining elements.  Let $\U(X,t)$ denote the uniform distribution over subsets of $X$ of size $t$.  Throughout the paper we always sample sets of size $t=k/r$ and therefore write $\mathcal{U}(X)$ instead of $\mathcal{U}(X,\frac{k}{r})$ to simplify notation.  The \algone \ algorithm is described formally above as Algorithm~\ref{alg:1}.     

 \begin{algorithm}[H]
\caption{\algone}
\begin{algorithmic}
    	\INPUT  constraint $k$, bound on number of iterations $r$
    	\STATE  $S \leftarrow \emptyset $
    	\STATE \textbf{for} $r$ \text{iterations}   \textbf{do}
    \STATE \ \ \ $X \leftarrow \textsc{Filter}(N,S,r)$
	\STATE \ \  \ $S \leftarrow S \cup R$, where $R \sim \U(X)$
    	\RETURN $S$ 
  \end{algorithmic}
  \label{alg:1}
\end{algorithm}

The \filter \ subroutine iteratively discards elements until a random set $R \sim \U(X)$ has marginal contribution arbitrarily close to the desired $(\OPT - f(S))/r$ value.  In each iteration, the elements discarded from the set of surviving elements $X$ are those whose marginal contribution to $R \sim \U(X)$ is low.  Intuitively, \filter \ terminates quickly since if a random set has low expected marginal contribution, then there are many elements whose marginal contribution to a random set is low and these elements are then discarded. The subroutine \filter \ is formally described below.

 \begin{algorithm}[H]
\caption{\textsc{Filter}$(X,S,r)$}
\begin{algorithmic}
    	\INPUT  Remaining elements $X$, current solution $S$, bound on number of outer-iterations $r$
   	\STATE  \textbf{while}  $\E_{R \sim \U(X)}\left[ f_S(R)\right]  <  \left(1 -  \epsilon \right) \left( \OPT - f(S) \right) / r$  \textbf{do} 
   			\STATE  \ \ \ $X \leftarrow X \setminus \left\{a \ : \ \E_{R \sim \U(X)} \left[f_{S \cup (R\setminus \{a\})}(a) \right] <  \left(1 + \epsilon/2 \right)\left(1 - \epsilon \right) \left( \OPT - f(S) \right) / k \right\}$
    	\RETURN $X$ 
  \end{algorithmic}
  \label{alg:filter}
\end{algorithm}

Both \algone \ and \filter \  are idealized versions of the algorithms we implement.  This is due to the fact that we do not know the value of the optimal solution $\OPT$ and we cannot compute expectations exactly. In practice, we can apply multiple guesses of $\OPT$ in parallel and estimate expectations by repeated sampling. For ease of presentation we analyze these idealized versions of the algorithms and defer the presentation and  analysis of the full algorithm to Appendix~\ref{sec:appFull}.  In our analysis we assume that in \algone \ when $\E_{R \sim \U(X)}[f_S(R)] \geq t$ this implies that a random set $R\sim \mathcal{U}(X)$ respects $f_S(R) \geq t$.\footnote{Since we estimate $E_{R \sim \U(X)}[f_S(R)]$ by sampling in the full version of the algorithm, there is at least one sample with value at least the estimated value of $E_{R \sim \U(X)}[f_S(R)]$ that we can take.}

\subsection{Analysis}  
The analysis of \algone \ relies on two properties of its \filter \ subroutine: (1) the marginal contribution of the set of elements not discarded in \filter \ after $\mathcal{O}(r)$ iterations is arbitrarily close to $(\OPT-f(S))/r$ and (2) there are at most $k/r$ remaining elements after $\mathcal{O}(r)$ rounds.  We assume that $\epsilon > 0$ is a small constant in the analysis.

\subsubsection{Bounding the value of elements that survive \textsc{Filter}}
\label{sec:alg1apx}

We first prove that the marginal contribution of elements returned by \filter \ to the existing solution $S$ is arbitrarily close to $(\OPT - f(S))/r$.  We do so by arguing that the set returned by \filter \ includes a subset of the optimal solution $O$ with such marginal contribution.  
Let $\rho$ be the number of iterations of the while loop in \filter.  For a given iteration $i \in [\rho]$ let $R_i$ be a random set of size $\frac{r}{k}$ drawn uniformly at random from $X_i$, where $X_i$ are the remaining elements at iteration $i$.  Notice that by monotonicity and submodularity, $f_{S}(O) \geq \OPT - f(S)$.  We first show that we can consider the marginal contribution of $O$ not only to $S$ but $S \cup (\cup_{i=1}^\rho R_i)$, while suffering an arbitrarily small loss. Considering the marginal contribution over random sets $R_i$ is important to show that some optimal elements of high value must survive all rounds.

 \begin{lemma} 
\label{lem:optContrib}
Let $R_i \sim \U(X)$ be the random set at iteration $i$ of \filter$(N,S,r)$. For all $S \subseteq N$ and $r, \rho > 0$, if \filter$(N,S,r)$ has not terminated after $\rho$ iterations, then
$$\EU[R_1, \ldots, R_\rho]\left[f_{S \cup \left(\cup_{i=1}^\rho R_i \right)}\left(O\right)\right] \geq \left(1 - \frac{\rho}{r}\right) \cdot \left(\OPT - f(S)\right).$$
\end{lemma}
\begin{proof} We exploit the fact that if \filter$(N,S,r)$  has not terminated after $\rho$ iterations, then by the algorithm, the random set $R_i \sim \U(X)$ at iteration $i$ has expected value that is upper bounded as follows:
$$\EU[R_i]\left[f_{S }\left(  R_i \right)\right] < \frac{1- \epsilon}{r} \left(\OPT - f(S)\right)$$
for all $i \leq \rho$. Next, by subadditivity, we have  
$\E_{R_1, \ldots, R_\rho}\left[f_{S }\left( \left(\cup_{i=1}^\rho R_i \right)\right)\right]  \leq \sum_{i=1}^\rho \E_{R_i}\left[f_{S }\left(  R_i \right)\right]$ and, by monotonicity, $\E_{R_1, \ldots, R_\rho}\left[f_{S }\left(O\cup \left(\cup_{i=1}^\rho R_i \right)\right)\right] \geq \OPT - f(S)$. Combining the above inequalities, we conclude that
\begin{align*}
\EU[R_1, \ldots, R_\rho]\left[f_{S \cup \left(\cup_{i=1}^\rho R_i \right)}\left(O\right)\right] & = \EU[R_1, \ldots, R_\rho]\left[f_{S }\left(O\cup \left(\cup_{i=1}^\rho R_i \right)\right)\right] -  \EU[R_1, \ldots, R_\rho]\left[f_{S }\left( \left(\cup_{i=1}^\rho R_i \right)\right)\right] \\
 & \geq \OPT - f(S)- \sum_{i=1}^\rho \EU[R_i]\left[f_{S }\left(R_i \right)\right]\\
 &  \geq 
\left(1 - \frac{\rho}{r}\right) \cdot \left(\OPT - f(S)\right). \qedhere
\end{align*} 
\end{proof}

 Next, we bound the value of elements that survive filtering rounds. To do so, we use Lemma~\ref{lem:optContrib} to show that there exists a a subset $T$ of the optimal solution $O$ that survives $\rho$ rounds of filtering and that has marginal contribution to $S$ arbitrarily close to $(\OPT - f(S))/r$.
\begin{lemma}
\label{lem:Tstar}
For all $S \subseteq N$ and $\epsilon > 0$, if $r \geq 20 \rho \epsilon^{-1}$, then the  elements $X_\rho$ that survive $\rho$ iterations  \filter$(N,S,r)$  satisfy 
$$f_S(X_{\rho}) \geq \frac{1}{r} \left(1 -  \epsilon \right) \left( \OPT - f(S) \right).$$
\end{lemma}
\begin{proof}
At a high level, the proof first defines a subset $T$ of the optimal solution $O$. Then, the remaining of the proof consists of two main parts. First, we show that elements in $T$ survive $\rho$ iterations of \filter$(N,S,r)$. Then, we show that $f_S(T) \geq \frac{1}{r} \left(1 -  \epsilon \right) \left( \OPT - f(S) \right).$ We introduce some notation. Let $O = \{o_1, \ldots, o_k\}$ be the optimal elements in some arbitrary order and $O_\ell = \{o_1, \ldots, o_\ell\}$. We define the following marginal contribution $\Delta_\ell$ of each optimal element $o_\ell$:
$$\Delta_\ell := \EU[R_1, \ldots, R_\rho] \left[ f_{S \cup O_{\ell-1}\cup \left(\cup_{i=1}^\rho R_i \setminus \{o_\ell\}\right) }(o_\ell)\right].$$
We define $T$ to be the set of optimal elements $o_\ell$ such that $\Delta_\ell \geq (1 - \epsilon/4) \Delta$ where
$$\Delta := \frac{1}{k}\left( 1 - \frac{\rho}{r}   \right) \cdot \left(\OPT - f(S)\right).$$
 We first argue that elements in $T$ survive $\rho$ iterations of \filter$(N,S,r)$. For element $o_\ell \in T$, we have
$$\Delta_\ell  \geq  (1 - \epsilon/4)\Delta  \geq \frac{1}{k}(1 - \epsilon/4)\left( 1 - \frac{\rho}{r}  \right) \cdot \left(\OPT - f(S)\right)  \geq \frac{1}{k}(1+\epsilon/2) (1 -  \epsilon)  \cdot \left(\OPT - f(S)\right)$$
where the last inequality is since since  $r \geq 20 \rho \epsilon^{-1}$. Thus, at iteration $i \leq \rho$, by submodularity, 
$$\EU[R_i  ] \left[f_{S \cup (R_i \setminus \{o_\ell\})}(o_\ell) \right] \geq  \EU[R_1, \ldots, R_\rho] \left[ f_{S \cup O_{\ell-1}\cup \left(\cup_{i=1}^\rho R_i  \setminus \{o_\ell\}\right)}(o_\ell)\right] = \Delta_\ell \geq \frac{1}{k}(1+\epsilon/2) (1 -  \epsilon)  \cdot \left(\OPT - f(S)\right) $$ 
and $o_\ell$ survives all iterations $i \leq \rho$, for all $o_\ell \in T$. 

Next, we argue that $f_S(T) \geq \frac{1}{r} \left(1 -  \epsilon \right) \left( \OPT - f(S) \right).$ Note that
$$\sum_{\ell =1}^k \Delta_\ell \geq   \EU[R_1, \ldots, R_\rho]\left[f_{S \cup \left(\cup_{i=1}^\rho R_i \right)}\left(O\right)\right] \geq\left( 1 - \frac{\rho}{r}  \right) \cdot \left(\OPT - f(S)\right) =  k \Delta.$$

where the second inequality is by Lemma~\ref{lem:optContrib}. Next, observe that

$$ \sum_{\ell =1}^k \Delta_\ell = \sum_{o_\ell \in T} \Delta_\ell + \sum_{j  \in O \setminus T} \Delta_\ell \leq \sum_{o_\ell \in T} \Delta_\ell + k (1 - \epsilon/4)\Delta.$$
By combining the two inequalities above, we get  $\sum_{o_\ell \in T} \Delta_\ell  \geq k \epsilon \Delta/4$. Thus, by submodularity,
\begin{align*}
f_S(T)  \geq \sum_{o_\ell \in T} f_{S \cup O_{\ell-1}}\left(o_\ell\right) 
 \geq \sum_{o_\ell \in T}\EU[R_1, \ldots, R_\rho] \left[ f_{S \cup O_{\ell-1}\cup \left(\cup_{i=1}^\rho R_i  \setminus \{o_\ell\}\right)}(o_\ell)\right]  
 = \sum_{o_\ell \in T} \Delta_\ell  
 \geq k \epsilon \Delta/4. 
\end{align*}
We conclude that 
\begin{align*}
f_S(X_{\rho}) \geq f_S(T)  \geq k \epsilon \Delta/4 = (\epsilon/4) \left( 1 - \frac{\rho}{r}  \right) \cdot \left(\OPT - f(S)\right) 
\geq \frac{1}{r} \cdot  (1 -  \epsilon)  \cdot \left(\OPT - f(S)\right).
\end{align*}
where the first inequality is by monotonicity and since $T \subseteq X_\rho$ is  a set of surviving elements.
\end{proof}

\subsection{The adaptivity of \textsc{Filter}} 
\label{sec:alg1rounds}

The second part of the analysis bounds the number of adaptive rounds of the \filter \ algorithm. A main lemma for this part, Lemma~\ref{lem:pruning}, shows that a constant number of elements are discarded at every round of filtering. Combined with the previous lemma that bounds the value of remaining elements, Lemma~\ref{lem:rounds} then shows that \filter \ has at most $\log n$ rounds. The analysis that a constant number of elements are discarded at every round is similar as in~\cite{BS18} and we defer the proof to the appendix for completeness. Since this is an important lemma, we give a proof sketch nevertheless.

\begin{restatable}{rLem}{lempruning}
\label{lem:pruning}
Let $X_i$ and $X_{i+1}$ be the surviving elements at the start and end of iteration $i$ of \filter$(N,S,r)$. For all $S \subseteq N$ and $r, i, \epsilon > 0$, if \filter$(N,S,r)$ does not terminate at iteration $i$, then
$$|X_{i+1}| < \frac{|X_i|}{1+\epsilon/2}.$$
\end{restatable}
\begin{proof}[Proof Sketch (full proof in Appendix~\ref{sec:appalg1})]
At a high level, since the surviving elements must have high value and a random set has low value, we can then use the thresholds to bound how many such surviving elements there can be while also having a random set of low value. To do so, we focus on the value of $f(R_i \cap X_{i+1})$ of the surviving elements $X_{i+1}$ in a random set $R_i \sim \D_{X_i}$.

First, the proof uses submodularity and  the threshold for elements in $X_{i+1}$ to survive from $X_i$ to show that $\E \left[f_S(R_i \cap X_{i+1})\right]  \geq |X_{i+1}| \cdot \frac{1}{r|X_i|} \cdot\left(1 + \epsilon/2 \right)\left(1 -  \epsilon \right) \left( \OPT - f(S) \right)$. Using monotonicity and the bound on the value of a random set $\E \left[f_S(R_i)\right]$ for \filter \ to discard additional elements, the proof then shows that $   \E \left[f_S(R_i \cap X_{i+1})\right] <  \frac{1}{r} \left(1 - \epsilon \right) \left( \OPT - f(S) \right)$ and concludes that $|X_{i+1}| \leq |X_i|/ (1+\epsilon/2)$.
\end{proof}

Thus, by the previous lemma, there are at most $k/r$ surviving elements after logarithmically many filtering rounds and by Lemma~\ref{lem:Tstar}, these remaining elements must have high value. Thus, \filter \ terminates and we obtain the following main lemma for the number of rounds.

\begin{lemma}
\label{lem:rounds}
For all $S \subseteq N$, if $r \geq 20 \epsilon^{-1}\log_{1+\epsilon/2}(n)  $, then  \filter$(N,S,r)$ terminates after at most $\O\left(\log n\right)$ iterations.
\end{lemma}

\begin{proof} 
If \filter$(N,S,r)$ has not yet terminated after $\log_{1+\epsilon/2}(n)$ iterations, then,
by Lemma~\ref{lem:pruning}, at most $k/r$ elements survived these $\rho = \log_{1+\epsilon/2}(n)$ iterations. By Lemma~\ref{lem:Tstar}, with $r \geq 20 \rho \epsilon^{-1} $, the set $X_\rho$ of elements that survive these $\log_{1+\epsilon/2}(n)$ iterations is such that $f_S(X_{\rho}) \geq \frac{1}{r} \cdot \left(1 -  \epsilon \right) \left( \OPT - f(S) \right)$. Since there are at most $k/r$ surviving elements $X$, $R = X_{\rho}$ for $R \sim \U(X_{\rho})$ and 
$$f_S(R) = f_S(X_{\rho}) \geq \frac{1}{r} \cdot \left(1 - \epsilon \right) \left( \OPT - f(S) \right),$$
and \filter$(N,S,r)$ terminates at this iteration.
\end{proof}

\paragraph{Main result for \algone.}
We are now ready to prove the main result for \algone. By Lemma~\ref{lem:rounds}, at every iteration of \algone, in at most $\O\left(\log n\right)$ iterations of \filter, the value of the solution $S$ is increased by at least $\left(1 -  \epsilon \right) \left( \OPT - f(S) \right) / r$  with $k/r$ new elements. The analysis of the $1 -1/e-\epsilon$ approximation then follows similarly as for the standard analysis of the greedy algorithm. Regarding the total number of rounds, we fix parameter $r = 20 \epsilon^{-1}\log_{1+\epsilon/2}(n) $. there are at most $r $ iterations of \algone, each of which with at most $\O\left(\log n\right)$ iterations of \filter \  and the queries at every iteration of \filter \ are non-adaptive. We defer the proof to Appendix~\ref{sec:appalg1}.

\begin{restatable}{rThm}{thmalgone}
\label{thm:algone}
For any constant $\epsilon > 0$, \algone  \ is a  $\mathcal{O}\left(\log^2 n\right)$-adaptive algorithm that obtains a $1 - 1/e - \epsilon$ approximation, with parameter $r = 20 \epsilon^{-1} \log_{1+\epsilon/2}(n) $.
\end{restatable}


\section{\algtwo : An $\mathcal{O}(\log n)$-adaptive Algorithm}\label{sec:alg2}
In this section, we build on the algorithm and analysis from the previous section to obtain the main result of this paper. We present \algtwo \ which accelerates \algone \ by using less filtering rounds while maintaining the same approximation guarantee.  In particular, it obtains an approximation arbitrarily close to $1-1/e$ in logarithmically-many adaptive rounds.

\subsection{Description of the algorithm}

\algtwo \  iteratively adds a block of $k/r$ elements obtained using the \filter \ subroutine to the existing solution $S$, exactly as \algone. The improvement in adaptivity comes from the use of \emph{epochs}. An epoch is a sequence of iterations during which the value of the solution $S$ increases by at most $\epsilon(\OPT - f(S))/20$. During an epoch, the algorithm invokes \filter \ with the surviving elements from the previous iteration of \algtwo, rather than all elements in the ground set as in \algone. In a new epoch, \filter \ is then again invoked with the ground set. A formal description of an idealized version is included below.

 \begin{algorithm}[H]
\caption{\algtwo}
\begin{algorithmic}
    	\INPUT  bound on number of iterations  $r$
    	\STATE $S \leftarrow \emptyset$
    	\STATE \textbf{for}  $\frac{20}{\epsilon}$ epochs   \textbf{do}
    	\STATE \ \ \  $ X \leftarrow N, T \leftarrow \emptyset$
    	\STATE \ \ \ \textbf{while}  $f_{S}(T) <  (\epsilon/20)(\OPT - f(S))$ and $|S \cup T| < k$ \textbf{do}
    \STATE \ \ \ \ \ \  $X \leftarrow \textsc{Filter}(X,S \cup T,r)$
	\STATE \ \  \ \ \ \ $T \leftarrow T \cup R$, where $R \sim \mathcal{U}(X)$
    	\STATE \ \ \  $S \leftarrow S \cup T$
    	\RETURN $S$ 
  \end{algorithmic}
  \label{alg:2}
\end{algorithm}

\subsection{Analysis of \algtwo} 
As in the previous section, we analyze the idealized version described above and defer the analysis of the full algorithm to the appendix.  Our analysis for \algtwo \ relies on the properties of every \emph{epoch}. In particular, we first show that during an epoch, the surviving elements $X$ have marginal contribution at least $\epsilon(\OPT - f(S))/20$ to $S \cup T$ (Section~\ref{sec:alg2apx}). Notice that the marginal contribution is with respect to the set $S \cup T$ and the value with respect only to $S$.  We then show that for any epoch, the total number of iterations of \filter \ during that epoch is $\mathcal{O}(\log n )$  (Section~\ref{sec:alg2rounds}). We emphasize that an iteration of \filter \ is different than an iteration of \algtwo, i.e., an epoch consists of multiple iterations of \algtwo, each of which consists of multiple iterations of \filter.  Since there are at most $20/\epsilon$ epochs, the amortized number of iterations of \filter \ per iteration of \algtwo \ is now constant.

\subsubsection{Bounding the value of elements that survive an epoch}
\label{sec:alg2apx}

For any given epoch, we first bound the marginal contribution of $O$ to $S\cup T$ and the random sets $\{R_i\}_{i=1}^{\rho}$ when there are $\rho$ iterations of filtering during the epoch.  Similar to the previous section, we show that the marginal contribution of $O$ to $S \cup T$ and the random sets is arbitrarily close to the desired $\OPT - f(S)$ value. The analysis is similar to the analysis of Lemma~\ref{lem:optContrib}, except for a subtle yet crucial difference.  The analysis in this section needs to handle the fact that the solution $S \cup T$ changes during the epoch. To do so we rely on the fact that the increase in the value of $S \cup T$ during an epoch is bounded.  Due to space considerations, we defer the proof to Appendix~\ref{sec:appalg2}.

 \begin{restatable}{rLem}{optContribTwo} 
\label{lem:optContrib2} For any epoch $j$ and $\epsilon > 0$,
let $R_i \sim \U(X)$ be the random set at iteration $i$ of filtering during epoch $j$. For all $r, \rho > 0$, if epoch $j$ has not ended after $\rho$ iterations of filtering, then
$$\EU[R_1, \ldots, R_\rho]\left[f_{S_j^+ \cup \left(\cup_{i=1}^\rho R_i \right)}\left(O\right)\right] \geq \left(1   - \frac{\rho}{r} - \epsilon/20\right) \cdot \left(\OPT - f(S_j)\right)$$
where $S_j$ is the set $S$ at epoch $j$ and $S_j^+$ is the set $S \cup T$ at the last iteration of epoch $j$.
\end{restatable}

Next, we bound the value of elements that survive the filtering iterations during an epoch. The proof is similar to that of Lemma~\ref{lem:Tstar}, modified to handle the fact that the solution $S$ evolves during an epoch. We defer the proof to Appendix~\ref{sec:appalg2}.

\begin{restatable}{rLem}{lemTstartwo}
\label{lem:Tstar2}
For any epoch $j$ and  $\epsilon > 0$, if $r \geq 20 \rho \epsilon^{-1}$, then the  elements $X_\rho$ that survive $\rho$ iterations of filtering at epoch $j$  satisfy 
$$f_{S_j^+}(X_{\rho}) \geq (\epsilon/4) \left(1 - \epsilon \right) \left( \OPT - f(S_j) \right).$$
where $S_j$ is the set $S$ at epoch $j$ and $S_j^+$ is the set $S \cup T$ at the last iteration of epoch $j$.
\end{restatable}

\subsubsection{The adaptivity of an epoch} \label{sec:alg2rounds}

The next lemma bounds  the total number of iterations of filtering per epoch. At a high level,
similarly as for \algone, a constant number of elements are discarded at each iteration of filtering by Lemma~\ref{lem:pruning} and there are at most $k/r$ surviving elements after logarithmically-many filtering rounds. Then, we use Lemma~\ref{lem:Tstar2} and the fact that the surviving elements during an epoch have high contribution to show that the epoch terminates. 

\begin{lemma}
\label{lem:rounds2}
In any epoch  of \algtwo \ and for any $\epsilon \in (0, 1/2)$, if $r \geq 20 \log_{1+\epsilon/2}(n)/ \epsilon$, then  there are at most $\log_{1+\epsilon/2}(n)$ iterations of filtering during the epoch.
\end{lemma}
\begin{proof}
If an epoch $j$  has not yet terminated after $\rho = \log_{1+\epsilon/2}(n)$ iterations of filtering, then,
by Lemma~\ref{lem:pruning}, at most $k/r$ elements survived these $\rho$ filtering iterations. We consider the set $T$ obtained after these $\rho$ filtering iterations.  By Lemma~\ref{lem:Tstar2}, with  $r \geq 20 \rho\cdot  \epsilon^{-1}$, the set $X_{\rho}$ of elements that survive these iterations is such that $f_{S \cup T}(X_{\rho}) \geq (\epsilon/4) \cdot \left(1 - \epsilon \right) \left( \OPT - f(S) \right)$. Since there are at most $k/r$ surviving elements, $R = X_\rho$ for $R \sim \U(X_\rho)$ and 
$$\E\left[f_{S \cup T}(R)\right] \geq f_{S \cup T}(X_{\rho}) \geq (\epsilon/4) \cdot \left(1 - \epsilon \right) \left( \OPT - f(S) \right) \geq \frac{1}{r}\cdot \left(1 - \epsilon \right) \left( \OPT - f(S \cup T) \right)$$
where the last inequality is by monotonicity. Thus, the current call to the \filter \ subroutine terminates and $X_{\rho}$ is added to $T$ by the algorithm. Next,
$$f_{S}(T \cup X_{\rho}) \geq f_{S}( X_{\rho}) \geq  f_{S \cup T}( X_{\rho}) \geq  (\epsilon/4) \cdot \left(1 - \epsilon \right) \left( \OPT - f(S) \right) \geq (\epsilon/20)\left( \OPT - f(S) \right) $$
the first inequality is by monotonicity and the second by submodularity. Thus,  epoch $j$ ends.
\end{proof}

\subsection{Main result}

We are now ready to prove the main result of the paper which is that analysis of  \algtwo. There are two cases: either the algorithm terminates after $r$ iterations with $|S \cup T| = k$ or it terminates after $20/\epsilon$ epochs. 
With $r = \mathcal{O} ( \log n )$, there are at most $ \mathcal{O}( \log n )$ iterations of adding elements and at most $\mathcal{O}(1)$ epochs with $\mathcal{O} (\log n )$ filtering iterations per epoch. Thus the total number of adaptive rounds is  $\mathcal{O}(\log n )$.
The proof is deferred to Appendix~\ref{sec:appalg2}
\begin{restatable}{rThm}{thmalgtwo}
\label{thm:alg2}
For any constant $\epsilon > 0$, when using parameter $r = 20\epsilon^{-1}\log_{1+\epsilon/2} (n)$,  \algtwo  \ obtains a $1 - 1/e - \epsilon$ approximation in  $\mathcal{O} (\log n)$-adaptive  steps.
\end{restatable}

Similarly as for \algone, \algtwo \ is an idealized version of the full algorithm since we do not know $\OPT$ and cannot compute expectations exactly. The full algorithm, which guesses $\OPT$ and estimates expectations arbitrarily well by sampling in one adaptive round, is formally described and analyzed in Appendix~\ref{sec:appFull}. The algorithm is randomized due to the sampling at every round and its analysis is nearly identical to that presented in this section while accounting for an additional arbitrarily small errors due to the guessing of $\OPT$ and the estimates of the expectation. The main result is the following theorem for the full algorithm.

\begin{restatable}{rThm}{thmalgFull}
\label{thm:algFull}
For any $\epsilon \in (0, 1/2)$, there exists an algorithm that obtains a $1 - 1/e - \epsilon$ approximation with probability $1 - \delta$ in $\mathcal{O}({\epsilon^{-1} \log_{1+\epsilon/3} n })$ adaptive steps.  Its query complexity in each round is $\mathcal{O}\left(n (  k + \log_{1+\epsilon/3} n  )^2 \frac{1}{\epsilon^2}\log\left(\frac{n}{\delta}\right) \log_{(1 - \epsilon/20)^{-1}}\left(n\right)\right)$.
\end{restatable}

\newpage

\bibliographystyle{alpha}
\bibliography{biblio}

\newpage

\appendix

\section*{Appendix}


\section{Additional Discussion of Related Work}

\subsection{Adaptivity}
\label{sec:adaptivity}

Adaptivity has been heavily studied across a wide spectrum of areas in computer science.  These areas include classical problems in theoretical computer science such as sorting and selection (e.g. \cite{Val75, Col88, BMW16}), where adaptivity is known under the term of parallel algorithms, and communication complexity (e.g. \cite{papadimitriou1984communication, duris1984lower, nisan1991rounds,miltersen1995data,dobzinski2014economic,alon2015welfare}), where the number of rounds measures how much interaction is needed for a communication protocol.

 For the  multi-armed bandits problem, the relationship of interest is between  adaptivity and query complexity, instead of adaptivity and approximation guarantee. Recent work  showed that $\Theta(\log^{\star} n)$ adaptive rounds are necessary and sufficient to obtain  the optimal worst case query complexity \cite{AAAK17}. In the bandits setting, adaptivity is necessary to obtain non-trivial query complexity due to the noisy outcomes of the queries. In contrast, queries in submodular optimization are deterministic and adaptivity is necessary to obtain a non trivial approximation  since there are at most polynomially many queries per round and the function is  of exponential size.  Adaptivity is also well-studied for the problems of sparse recovery (e.g. \cite{HNC09,IPW11, haupt2009compressive,ji2008bayesian,malioutov2008compressed,aldroubi2008sequential})
 and  property testing (e.g. \cite{CG17, buhrman2012non,chen2017settling,raskhodnikova2006note,servedio2015adaptivity}). In these areas, it has been shown that
 adaptivity allows significant improvements compared to the non-adaptive setting, which is similar to the results shown in this paper for submodular optimization. However, in contrast to all these areas, adaptivity has not been previously studied in the context of submodular optimization.

We note that the term adaptive submodular maximization has been  previously used, but in an unrelated setting where the goal is to compute a policy which iteratively picks elements one by one, which, when picked, reveal stochastic feedback about the environment \cite{golovin2010adaptive}.

\subsection{Related models of parallelism}
\label{sec:apprelatedparallelism}

\subsubsection{Parallel computing and depth}
\label{sec:appdepth}

Our main result extends to the PRAM model. Let $d_f$ be the depth required to evaluate the function on a set, then there is a $\tilde{\mathcal{O}}(\log^2 n\cdot d_f)$ depth algorithm with $\tilde{\mathcal{O}}(nk^2)$ work whose approximation is arbitrarily close to $1-1/e$ for submodular maximization under a cardinality constraint.

The PRAM model is a generalization of the RAM model with parallelization, it is an idealized model of a shared memory machine with any number of processors which can execute instructions in parallel. The \emph{depth} of a PRAM algorithm  is the longest chain of dependencies of the algorithm, including operations which are not necessarily queries. Thus,  in addition to the number of adaptive rounds of querying, depth also measures the number of adaptive steps of the algorithms which are not queries. 
The additional factor in the depth  compared to the number of adaptive rounds is $d_f \cdot \tilde{O}(\log n)$ , where $d_f$ is the depth required to evaluate the function on a set in the PRAM model. The operations that our algorithms performed at every round, which are maximum, summation, set union, and set difference over an input of size at most quasilinear,  can all be executed by algorithms with logarithmic depth. A simple divide-and-conquer approach suffices for maximum and summation, while logarithmic depth for set union and set difference can be achieved with treaps \cite{blelloch1998fast}.

\subsubsection{Map-Reduce}
\label{sec:appmapreduce}

The problem of distributed  submodular optimization has been extensively studied in the Map-Reduce model in the past decade. This framework is primarily motivated by large scale problems over massive data sets. At a high level, in the Map-Reduce framework \cite{dean2008mapreduce}, an algorithm proceeds in multiple Map-Reduce rounds, where each round consists of a first step where the input to the algorithm is partitioned to be independently processed  on different machines and of a second step where  the outputs of this processing are merged.  Notice that the notion of rounds in Map-Reduce is different than for adaptivity, where one round of Map-Reduce usually consists of multiple adaptive rounds. The formal model of \cite{karloff2010model} for Map-Reduce requires the number of machines and their memory to be sublinear.

This framework for distributing the input to multiple machines with sublinear memory is designed to tackle issues related to massive data sets. Such data sets are too large to either  fit or be processed by a single  machine and the Map-Reduce framework formally models this need to distribute such inputs to multiple machines.

Instead of addressing distributed challenges, adaptivity addresses the issue of sequentiality, where each query evaluation requires a long time to complete and where these evaluations can be parallelized (see Section~\ref{sec:appapplications} for applications).  In other words,  while Map-Reduce addresses the \emph{horizontal}  challenge of large scale problems, adaptivity addresses an  orthogonal \emph{vertical} challenge where  long query-evaluation time is causing the main runtime bottleneck.

A long line of work  has  studied problems related to submodular maximization in Map-Reduce achieving different improvements on parameters such as the number of Map-Reduce rounds, the communication complexity, the approximation ratio, the family of functions, and the family of constraints (e.g. \cite{ kumar2015fast, mirzasoleiman2013distributed, mirrokni2015randomized, mirzasoleiman2015distributed,barbosa2015power, barbosa2016new,EMZ17}). To the best of our knowledge, all the existing Map-Reduce algorithms for submodular optimization have adaptivity that is linear in $n$ in the worst-case, which is  exponentially larger than the adaptivity of our algorithm. This high adaptivity is caused by the distributed algorithms which are run on each machine. These algorithms are variants of the greedy algorithm and thus have adaptivity at least linear in $k$. We also note that our algorithm does not (at least trivially) carry over to the Map-Reduce setting.

\section{Applications}
\label{sec:appapplications}

We discuss in detail several  applications of submodular optimization where sequentiality is the main runtime bottleneck.
In \textbf{crowdsourcing and data summarization}, 
algorithms involve subtasks performed by the crowd. The intervention of humans in the evaluation of queries  causes  algorithms with a large number of adaptive rounds to be impractical.  A crowdsourcing platform consists of posted tasks and crowdworkers who are remunerated for performing these posted tasks. For several submodular optimization problems, such as data summarization, the value of queries can be evaluated on a crowdsourcing platform \cite{tschiatschek2014learning,singla2016noisy,BMW16}. The algorithm must wait to obtain the feedback from the crowdworkers, however an algorithm can ask different crowdworkers to evaluate a large number of queries in parallel.

 In \textbf{experimental design}, the goal is to pick a collection of entities (e.g. subjects, chemical elements, data points) which obtains the best outcome when combined for an experiment. Experiments can be run in parallel and have a waiting time to observe the outcome \cite{frazier2010use}. The submodular problem of \textbf{influence maximization}, initiated studied by~\cite{domingos2001mining, richardson2002mining, kempe2003maximizing} has since then been well-studied (e.g. \cite{chen2009efficient,chen2010scalable,goyal2011celf++,seeman2013adaptive,horel2015scalable,badanidiyuru2016locally}). Influence maximization consists of finding the most influential nodes in a social network to maximize the spread of information in this network. 
Information does not spread instantly and an algorithm must wait to observe the total number of nodes influenced by some seed set of nodes. In \textbf{advertising}, the goal is to select the optimal subset of advertisement slots to objectives such as the click-through-rate  or the number of products purchased by customers, which are objectives exhibiting diminishing returns \cite{alaei2010maximizing,devanur2016whole}. Naturally, a waiting time is incurred to observe the performance of different collections of advertisements.

\subsection{Previous work on adaptivity for submodular maximization}
\label{sec:appcomparison}

The main algorithm in \cite{BS18}, \algadaptive, obtains a constant factor approximation in $\O(\log n)$ adaptive rounds. It consists of two primitives, \algdown \ and \algup. \algdown \ is $\O(\log n/ \log \log n)$-adaptive but only obtains a $\O(\log n)$ approximation. On the other hand, \algup \ obtains a constant factor approximation but in linearly many rounds. The main algorithm appropriately combines both primitives to obtain a constant factor approximation guarantee in $\O(\log n)$ rounds.

The main algorithm in this paper, \algtwo, mimics the greedy analysis to obtain an approximation arbitrarily close to $1-1/e$ by finding a block of size $k/r$ whose marginal contribution is arbitrarily close to $(\OPT - f(S))/r$. We first give \algone \ which finds such a set in $\O(\log n)$ rounds by filtering elements at every iteration. We build on that algorithm to obtain \algtwo, which obtains a $1-1/e$ approximation and   uses a concept of epoch to obtain an amortized number of rounds that is constant per iteration during an epoch. The analysis for the approximation is thus very different to obtain the $1-1/e$ approximation. One similarity is Lemma~\ref{lem:pruning} which shows that a constant number of elements can be discarded in one round, similarly as Lemma 1 in \cite{BS18}.

\section{Missing Analysis from Section~\ref{sec:alg1}}
\label{sec:appalg1}

\lempruning*
\begin{proof}
At a high level, since the surviving elements must have high value and a random set has low value, we can then use the thresholds to bound how many such surviving elements there can be while also having a random set of low value. To do so, we focus on the value of $f(R_i \cap X_{i+1})$ of the surviving elements $X_{i+1}$ in a random set $R_i \sim \D_{X_i}$.
\begin{align*}
\E \left[f_S(R_i \cap X_{i+1})\right] 
& \geq \E\left[\sum_{a \in R_i \cap X_{i+1}} f_{S \cup (R_i \cap X_{i+1} \setminus a)}(a)\right] & \text{submodularity}\\
& \geq \E\left[\sum_{a \in X_{i+1}} \mathds{1}_{a \in R_i} \cdot  f_{S \cup (R_i \setminus a)}(a)\right] & \text{submodularity}\\
& = \sum_{a \in X_{i+1}} \E\left[\mathds{1}_{a \in R_i} \cdot  f_{S \cup (R_i \setminus a)}(a)\right]. & \\
 & =\sum_{a \in X_{i+1}} \Pr\left[a \in R_i\right] \cdot \E\left[ f_{S \cup (R_i \setminus a)}(a) | a \in R_i\right] & \\
& \geq\sum_{a \in X_{i+1}} \Pr\left[a \in R_i\right] \cdot \E\left[ f_{S \cup (R_i \setminus a)}(a) \right] & \text{submodularity}\\
& \geq \sum_{a \in X_{i+1}} \Pr\left[a \in R_i\right] \cdot \frac{1}{k}\left(1 + \epsilon/2 \right)\left(1 - \epsilon \right) \left( \OPT - f(S) \right)  & \text{algorithm} \\
&  =  |X_{i+1}| \cdot \frac{k}{r|X_i|} \cdot \frac{1}{k} \left(1 + \epsilon/2 \right)\left(1 - \epsilon \right) \left( \OPT - f(S) \right) & \text{definition of $\U(X)$} \\
&  =  |X_{i+1}| \cdot \frac{1}{r|X_i|} \cdot\left(1 + \epsilon/2 \right)\left(1 - \epsilon \right) \left( \OPT - f(S) \right). & 
\end{align*}
Next, since elements are discarded, a random set must have low value by the algorithm,
$$ \frac{1}{r} \left(1 - \epsilon \right) \left( \OPT - f(S) \right) > \E \left[f_S(R_i)\right] .$$
By monotonicity, we get $\E \left[f_S(R_i)\right] \geq \E \left[f_S(R_i \cap X_{i+1})\right]$. Finally, by combining the above  inequalities, we  conclude that $|X_{i+1}| \leq |X_i|/ (1+\epsilon/2)$. 
\end{proof}

\thmalgone*
\begin{proof}
Let $S_i$ denote the solution $S$ at the $i$th iteration of \algone.
The algorithm increases the value of the solution $S$ by at least $\left(1 -  \epsilon \right) \left( \OPT - f(S) \right) / r$ at every iteration with $k/r$ new elements. Thus,
$$f(S_i) \geq f(S_{i-1}) + \frac{1- \epsilon}{r} \left(\OPT - f(S_{i-1})\right).$$
Next, we show by induction on $i$ that
$$f(S_i) \geq \left( 1 - \left(1 - \frac{1-\epsilon}{r}\right)^i\right) \OPT.$$
Observe that
\begin{align*}
f(S_i) & \geq f(S_{i-1}) + \frac{1-  \epsilon}{r} \left(\OPT - f(S_{i-1})\right) \\
& = \frac{1-  \epsilon}{r} \OPT - \left(1- \frac{1- \epsilon}{r}\right) f(S_{i-1})\\
& \geq \frac{1- \epsilon}{r} \OPT - \left(1- \frac{1- \epsilon}{r}\right)\left( 1 - \left(1 - \frac{1- \epsilon}{r}\right)^{i-1}\right) \OPT\\
& = \left( 1 - \left(1 - \frac{1-  \epsilon}{r}\right)^i\right) \OPT
\end{align*}
Thus, with $i = r$ where there has been $r$ iterations of adding $k/r$ elements, we return solution $S$ such that
$$f(S) \geq \left( 1 - \left(1 - \frac{1-  \epsilon}{r}\right)^r\right) \OPT$$
and obtain $$f(S) \geq \left(1 - e^{-(1-  \epsilon)}\right)\OPT \geq \left(1 - \frac{1 + 2 \epsilon}{e}\right)\OPT \geq  \geq \left(1 - \frac{1}{e} - \epsilon \right)\OPT$$ where the second inequality is since $e^{x} \leq 1 + 2 x$ for $0 < x < 1$. The number of rounds is at most $r \log_{1+\epsilon/2}(n)$ since there are $r$ iterations of \algone, each of which with at most $\log_{1+\epsilon/2}(n)$ iterations of \filter \ by Lemma~\ref{lem:rounds}, with $r = O(\log_{1+\epsilon/2}(n))$.
\end{proof}
\section{Missing Analysis from Section~\ref{sec:alg2}}
\label{sec:appalg2}

We introduce some notation and terminology. We now call \emph{the iteration $i$ of filtering during epoch $j$} the $i^{\text{th}}$ iteration discarding elements inside of \filter \ since the beginning of epoch $j$, over the multiple invokations of \filter. An element \emph{survives $\rho$ iterations of \filter \ at epoch $j$} if it has not been discarded at iteration $i$ of filtering during epoch $j$, for all $i \leq \rho$. 
 Let $S_j$ denote the solution $S$ at epoch $j \in [20/\epsilon]$, $S_j^+$ denote $S_j \cup T$ during the last iteration of \algtwo \ at epoch $j$, i.e., the last $T$ such that $f_{S}(T) <  (\epsilon/20)(\OPT - f(S))$, and $S_{j,i}$ denote $S_j \cup T $ at the  iteration $i$ of filtering  during epoch $j$. Thus, for all $i_1 < i_2$, $$S_j \subseteq S_{j,i_1} \subseteq S_{j,i_2}  \subseteq S_j^+ \subseteq S_{j+1}$$
and $f(S_j^+) - f(S_j) < (\epsilon/20) (\OPT - f(S_j)).$.

\optContribTwo*
\begin{proof} Similarly as for Lemma~\ref{lem:optContrib}, we exploit the fact that if \filter$(N,S,r)$  has not terminated after $\rho$ iterations, then by the algorithm, the random set $R_i \sim \U(X)$ at iteration $i$ has low expected value. In addition, we also use the bound on the change in value of $S$ during epoch $j$: 
\begin{align*}
& \EU[R_1, \ldots, R_\rho]\left[f_{S_j^+ \cup \left(\cup_{i=1}^\rho R_i \right)}\left(O\right)\right] \\
 =& \EU[R_1, \ldots, R_\rho]\left[f_{S_j^+ }\left(O\cup \left(\cup_{i=1}^\rho R_i \right)\right)\right] -  \EU[R_1, \ldots, R_\rho]\left[f_{S_j^+ }\left( \left(\cup_{i=1}^\rho R_i \right)\right)\right] & \\
\geq &\OPT - f(S_j^+) -  \EU[R_1, \ldots, R_\rho]\left[f_{S_j^+ }\left( \left(\cup_{i=1}^\rho R_i \right)\right)\right] & \text{monotonicity} \\
\geq &\OPT - f(S_j)  - (\epsilon/20) \left(\OPT - f(S)\right)  - \EU[R_1, \ldots, R_\rho]\left[f_{S_j^+ }\left( \left(\cup_{i=1}^\rho R_i \right)\right)\right] & \text{same epoch} \\
\geq &(1 - \epsilon/20)\left(\OPT - f(S)\right) - \sum_{i=1}^\rho \EU[R_i]\left[f_{S_j^+ }\left(  R_i \right)\right] & \text{subadditivity} \\
\geq &(1 - \epsilon/20)\left(\OPT - f(S)\right) - \sum_{i=1}^\rho \EU[R_i]\left[f_{S_{j,i}}\left(  R_i \right)\right] & \text{submodularity} \\
\geq &(1 - \epsilon/20)\left(\OPT - f(S)\right)- \sum_{i=1}^\rho \frac{1-\epsilon}{r} \left(\OPT - f(S_{j})\right)] & \text{algorithm} \\
= &\left(1 - \epsilon/20  - \frac{\rho}{r}\right) \cdot \left(\OPT - f(S_j)\right). &
\end{align*}
\end{proof}

\lemTstartwo*
\begin{proof} Let $j$ be any epoch.
Similarly as for Lemma~\ref{lem:Tstar}, the proof  defines a subset $Q$ of the optimal solution $O$ and then shows show that elements in $Q$ survive $\rho$ iterations of filtering at epoch $j$ and show that $f_{S_j^+}(Q) \geq (\epsilon/4) \left(1 -  \epsilon \right) \left( \OPT - f(S_j) \right).$  We define the following marginal contribution $\Delta_\ell$ of each optimal element $o_\ell$:
$$\Delta_\ell := \EU[R_1, \ldots, R_\rho] \left[ f_{S_j^+ \cup O_{\ell-1}\cup \left(\cup_{i=1}^\rho R_i \setminus \{o_\ell\}\right) }(o_\ell)\right].$$

We define $Q$ to be the set of optimal elements $o_\ell$ such that $\Delta_\ell \geq (1 - \epsilon/4) \Delta$ where
$$\Delta := \frac{1}{k}\left(1 - \frac{\rho}{r} - \epsilon/20 \right) \cdot \left(\OPT - f(S_j)\right).$$
 We first argue that elements in $Q$ survive $\rho$ iterations of filtering at epoch $j$. For element $o_\ell \in Q$, we have
$$\Delta_\ell  \geq  (1 - \epsilon/4)\Delta  \geq \frac{1}{k}(1 - \epsilon/4)\left(1- \frac{\rho}{r} - \epsilon/20\right)\cdot \left(\OPT - f(S_j)\right)  \geq \frac{1}{k}(1+\epsilon/2) (1 - \epsilon)  \cdot \left(\OPT - f(S_j)\right)$$
where the third inequality is by the condition on $r$. Thus, at iteration $i \leq \rho$, by submodularity, 
$$\EU[R_i  ] \left[f_{S_{j,i} \cup (R_i \setminus \{o_\ell\})}(o_\ell) \right] \geq  \EU[R_1, \ldots, R_\rho] \left[ f_{S_j^+ \cup O_{\ell-1}\cup \left(\cup_{i=1}^\rho R_i  \setminus \{o_\ell\}\right)}(o_\ell)\right] = \Delta_\ell \geq \frac{1}{k}(1+\epsilon/2) (1 - \epsilon)  \cdot \left(\OPT - f(S_j)\right) $$ 
and $o_\ell$ survives all iterations $i \leq \rho$, for all $o_\ell \in Q$. 

Next, we argue that $f_{S_j^+}(Q) \geq (\epsilon/4) \left(1 -  \epsilon \right) \left( \OPT - f(S_j) \right).$ Note that
$$\sum_{\ell =1}^k \Delta_\ell \geq   \EU[R_1, \ldots, R_\rho]\left[f_{S_j^+ \cup \left(\cup_{i=1}^\rho R_i \right)}\left(O\right)\right] \geq\left( 1 - \frac{\rho}{r}  - \epsilon/20\right) \cdot \left(\OPT - f(S_j)\right) =  k \Delta.$$

where the second inequality is by Lemma~\ref{lem:optContrib2}. Next, observe that

$$ \sum_{\ell =1}^k \Delta_\ell = \sum_{o_\ell \in Q} \Delta_\ell + \sum_{j  \in O \setminus Q} \Delta_\ell \leq \sum_{o_\ell \in Q} \Delta_\ell + k (1 -\epsilon/4)\Delta.$$
By combining the two inequalities above, we get  $\sum_{o_\ell \in Q} \Delta_\ell  \geq k \epsilon \Delta/4$. Thus, by submodularity,
\begin{align*}
f_{S_j^+}(Q)  \geq \sum_{o_\ell \in Q} f_{S_j^+ \cup O_{\ell-1}}\left(o_\ell\right) 
 \geq \sum_{o_\ell \in Q}\EU[R_1, \ldots, R_\rho] \left[ f_{S_j^+ \cup O_{\ell-1}\cup \left(\cup_{i=1}^\rho R_i  \setminus \{o_\ell\}\right)}(o_\ell)\right]  
 = \sum_{o_\ell \in Q} \Delta_\ell  
 \geq k \epsilon \Delta /4. 
\end{align*}
We conclude that 
\begin{align*}
f_{S_j^+}(X_{\rho}) \geq f_{S_j^+}(Q)  \geq k \epsilon \Delta /4 = (\epsilon/4) \left( 1 - \frac{\rho}{r}  \right) \cdot \left(\OPT - f(S_j)\right) 
\geq 
(\epsilon/4) \cdot  (1 -  \epsilon)  \cdot \left(\OPT - f(S_j)\right).
\end{align*}
where the first inequality is by monotonicity and since $Q \subseteq X_\rho$ is  a set of surviving elements.
\end{proof}

\thmalgtwo*
\begin{proof}
First, consider the case where the algorithm terminates after $r$ iterations of adding elements to $S$. Let $S_i$ denote the solution $S$ at the $i$th iteration.
\algtwo \ increases the value of the solution $S$ by at least $\left(1 - \epsilon \right) \left( \OPT - f(S) \right) / r$  at every iteration with $k/r$ new elements. Thus,
$$f(S_i) \geq f(S_{i-1}) + \frac{1-\epsilon}{r} \left(\OPT - f(S_{i-1})\right)$$
and we obtain $f(S) \geq \left(1 - e^{-(1-\epsilon)}\right)\OPT\geq \left(1 - e^{-1} - \epsilon\right)\OPT$ similarly as for Theorem~\ref{thm:algone}. 

Next, consider the case where the algorithm terminated after $(\epsilon/20)^{-1}$ epochs. At every epoch $j$, the algorithm increases the value of the solution $S$ by $(\epsilon/20)(\OPT - f(S_j))$. Thus,
$$f(S_j) \geq f(S_{j-1}) + (\epsilon/20)\left(\OPT - f(S_{j-1})\right).$$
Similarly as in the first case, we get that after $(\epsilon/20)^{-1}$ epochs, $f(S) \geq (1 - e^{-1}) \OPT$.

The total number of rounds of adaptivity of \algtwo \  is at most $20\epsilon^{-1} \log_{1+\epsilon/2}(n) + (\epsilon/20)^{-1} \log_{1+\epsilon/2}(n)$ since there are at most $r = 20 \epsilon^{-1} \log_{1+\epsilon/2}(n)$ iterations of adding elements and at most $(\epsilon/20)^{-1}$ epochs with, by Lemma~\ref{lem:rounds2}, at most $\log_{1+\epsilon/2}(n)$ filtering iterations each. The queries at each filtering iteration are independent and can be evaluated in parallel.
\end{proof}

\section{The Full Algorithm}
\label{sec:appFull}

\subsection{Description of the full algorithm}

\subsubsection{Estimates of expectations in one round via sampling}
\label{sec:appestimates}

We show that the expected value of a random set and  the expected marginal contribution of elements to a random set can be estimated arbitrarily well in one round, which is needed for the \algone \ and \algtwo \ algorithms. Recall that $\U(X)$ denotes the uniform distribution over subsets of $X$ of size $k/r$. The values we are interested in estimating are $\E_{R \sim \U(X)}\left[f_S(R)\right]$ and $\E_{R \sim \U(X)}\left[f_{S \cup (R \setminus a)} (a)\right]$. We denote the corresponding estimates by $v_S(X )$ and $v_{S}(X,a)$, which are computed in Algorithms~\ref{alg:est1} and \ref{alg:est2}. These algorithms first sample $m$ sets from $\U(X)$, where $m$ is the sample complexity, then query the desired sets to obtain a random realization of  $f_S(R)$ and $f_{S \cup (R \setminus a)} (a)$, and finally averages the $m$ random realizations of these values.

\begin{algorithm}[H]
\caption{\estOne: Computes estimate $v_S(X)$ of $\E_{R \sim \U(X)}\left[f_S(R)\right]$.}
\begin{algorithmic}
    	\INPUT  set $X \subseteq N$,  sample complexity $m$. 
    	\STATE Query $f(S)$ and $f(S \cup R_i)$ for all samples $R_1, \ldots, R_m \iid \U(X)$
    	\RETURN $\frac{1}{m}\sum_{i = 1}^m f_S(R_i)$
  \end{algorithmic}
  \label{alg:est1}
\end{algorithm}

\begin{algorithm}[H]
\caption{\estTwo: Computes  estimate $v_{S}(X,a)$ of $\E_{R \sim \U(X)}\left[f_{S \cup (R \setminus a)} (a)\right]$.}
\begin{algorithmic}
    	\INPUT  set $X \subseteq N$,  sample complexity $m$, element $a \in N$. 
    	\STATE Query $f(S \cup (R_i \cup a))$ and $f(S \cup (R_i \setminus a))$  for all samples  $R_1, \ldots, R_m \iid \U(X)$
    	\RETURN $\frac{1}{m}\sum_{i = 1}^m  f_{S \cup (R_i \setminus a)}(a)$
  \end{algorithmic}
  \label{alg:est2}
\end{algorithm}

Using standard concentration bounds, the estimates computed by these algorithms are arbitrarily good for a sufficiently large sample complexity $m$. We state the version of Hoeffding's inequality which is used to bound the error of these estimates.
\begin{lemma}[Hoeffding's inequality]
\label{l:hoeffding}
Let $S_1, \dots, S_n$ be independent random variables with values in $[0,b]$. Let $S = \frac{1}{m}\sum_{i=1}^m S_i$. Then for any $\epsilon > 0$,
$$\Pr\left[|S - \E[S]| \geq \epsilon\right] \leq 2 e^{- 2 m \epsilon^2 /  b^2}.$$
\end{lemma}
 We are now ready to show that these estimates are arbitrarily good.

\begin{lemma}
\label{lem:concentration}  
Let $m = \frac{1}{2} \left( \frac{\OPT}{\epsilon}\right)^2 \log\left(\frac{2}{\delta}\right)$,  then for all $S, X \subseteq N$, and $a \in N$, with probability $1 - \delta$ over the samples $R_1, \ldots, R_m$,
 $$\left|v_{S}(X,a) - \EU[R \sim \U(X)]\left[f_{S \cup (R \setminus a)| a \in R} (a)\right]\right| \leq \epsilon \ \ \ \ \ \ \text{ and } \ \ \ \ \ \ \left|v_S(X) - \EU[R \sim \U(X)]\left[f_S(R)\right] \right| \leq \epsilon.$$
 Thus, with $m = n \left( \frac{\OPT}{\epsilon}\right)^2 \log\left(\frac{2n}{\delta}\right)$ total samples in one round, with probability $1- \delta$, it holds that $v_S(X)$ and 
$v_{S}(X,a)$, for all $a \in N$, are $\epsilon$-estimates.
\end{lemma}
\begin{proof}
Note that 
$$\E\left[v_S(X)\right] = \EU[R \sim \U(X)]\left[f_S(R)\right] \ \ \ \ \ \ \text{ and } \ \ \ \ \ \  \E\left[v_{S}(X,a)\right]= \EU[R \sim \U(X)]\left[f_{S \cup (R \setminus a)} (a)\right]$$
Since all queries are of size at most $k$, their values are all bounded by $\OPT$. Thus, by Hoeffding's inequality with $m = \frac{1}{2} \left( \frac{\OPT}{\epsilon}\right)^2 \log\left(\frac{2}{\delta}\right)$, we get 
$$\Pr\left[\left| v_S(X) - \EU[R \sim \U(X)]\left[f_S(R)\right]\right| \geq  \epsilon\right], \Pr\left[\left| v_{S}(X,a) - \EU[R \sim \U(X)]\left[f_{S \cup (R \setminus a)} (a)\right]\right| \geq \epsilon \right] \leq 2 e^{-\frac{ 2 m \epsilon ^2}{\OPT^2}}  \leq \delta   $$
for $\epsilon > 0$. Thus, with $m = n \left( \frac{\OPT}{\epsilon}\right)^2 \log\left(\frac{2n}{\delta}\right)$ total samples in one round, by a union bound over each of the estimates holding with probability $1 - \delta/n$ individually, we get that all the estimates hold simultaneously with probability $1- \delta$.
\end{proof}

We can now describe the (almost) full version of the main algorithm which uses these estimates. Note that we can force the algorithm to stop after any round   to obtain the desired adaptive complexity with probability $1$. In our analysis, the loss from the event that the algorithm is forced to stop when the desired adaptivity is reached is accounted for in the $\delta$ probability of failure of the approximation guarantee of the algorithm.

 \begin{algorithm}[H]
\caption{\algMainAlmostFull}
\begin{algorithmic}
    	\INPUT  bound on number of iterations $r$, sample complexity $m$, proxy $v^{\star}$
    	\STATE $S \leftarrow \emptyset$
    	\STATE \textbf{for}  $\frac{20}{\epsilon}$ epochs   \textbf{do}
    	\STATE \ \ \  $ X \leftarrow N, T \leftarrow \emptyset$
    	\STATE  \ \ \  $v_S\left(X\right) \leftarrow \estOne\left(X, m\right)$
    	\STATE \ \ \ \textbf{while}  $v_S\left(X\right) <  (\epsilon/20)(v^{\star} - f(S))$ and $|S \cup T| < k$ \textbf{do}
\STATE \ \ \ \ \ \  \textbf{for} $a \in X$  \textbf{do} 
    	\COMMENT{Non-adaptive loop}
    	\STATE \ \ \ \ \ \ \ \ \ $v_S\left(X, a\right) \leftarrow \estTwo\left(X\setminus S, m, a\right)$
    			\STATE \ \ \ \ \ \ $X \leftarrow X \setminus \left\{a \ : \ v_S\left(X, a\right) <  \left(1 + \epsilon/2 \right)\left(1 -  \epsilon \right) \left(v^{\star} - f(S) \right) / k \right\}$
	\STATE \ \  \ \ \ \ $T \leftarrow T \cup R$, where $R \sim \mathcal{U}(X)$
	\STATE \ \ \ \ \ \  $v_S\left(X\right) \leftarrow \estOne\left(X, m\right)$
    	\STATE \ \ \  $S \leftarrow S \cup T$    	
    	\RETURN $S$ 
  \end{algorithmic}
  \label{alg:adaptive}
\end{algorithm}

\subsubsection{Estimates of \OPT}

The main idea to estimate $\OPT$ is to have $O(\log_{1 +\epsilon} n)$ values $v_i$ such that  one of them is guaranteed to be a $(1-\epsilon)$-approximation to $\OPT$.  To obtain such values, we use the simple observation that the singleton $a^{\star}$ with largest value is at least a $1/n$ approximation to $\OPT$. 

More formally, let  $a^{\star} = \argmax_{a \in N}f(a)$ be the optimal singleton, and $v_i = (1+\epsilon)^i \cdot  f\left(a^{\star}\right).$ We argue that there exists some $i \in \left[\log_{1+\epsilon} n\right]$ such that
$\OPT \leq v_i \leq (1+\epsilon) \cdot  \OPT.$ 
By submodularity, we get $f(a^{\star}) \geq \frac{1}{k} \OPT \geq \frac{1}{n} \OPT$. By monotonicity, we have $f(a^{\star}) \leq \OPT$. Combining these two inequalities, we get 
$v_0 \leq \OPT \leq v_{\log_{1+\epsilon} n}.$
By the definition of $v_i$, we then conclude that there must exists some $i \in \left[\log_{1+\epsilon} n\right]$ such that $\OPT \leq v_i \leq (1+\epsilon) \cdot  \OPT.$

Since the solution obtained for the unknown $v_i$ which approximates $\OPT$ well is guaranteed to be  a good solution, we run the algorithm in parallel for each of these values and return the solution with largest value. We obtain the full algorithm \algMainFull \ which we describe next.
 \begin{algorithm}[H]
\caption{\algMainFull}
\begin{algorithmic}
    	\INPUT  bounds  on number of outer-iterations $r$,  sample complexity $m$,  and precision  $\epsilon$
    	\STATE Query $f(a_1), \ldots, f(a_n)$
    	\STATE $a^{\star} \leftarrow \argmax_{a_i} f(a_i)$
    	\STATE \textbf{for} $i \in \left\{0, \ldots, \log_{1 + \epsilon}(n)\right\}$  \textbf{do} 
    	\COMMENT{Non-adaptive loop}
    	\STATE \ \ \ $v^{\star} \leftarrow (1+ \epsilon)^i \cdot f\left(a^\star\right)$
    	\STATE \ \ \ $X_i \leftarrow $ \algMainAlmostFull$(v^{\star})$ 
    	\RETURN best solution $X_i$:   $\argmax_{X_i : i \in \log_{1 + \epsilon}(n)} f(X_i)$
  \end{algorithmic}
\end{algorithm}

\subsection{Analysis of the \algMainFull \ algorithm}

We bound the number of elements removed from $X$ in each round of the full algorithm.  
\begin{lemma}
\label{lem:pruningFull} Assume $(1 - \epsilon/20) \OPT \leq v^{\star} \leq  \OPT$ and $0 < \epsilon< 1/2$.
For any $S$ and $r$, at the iteration $i$ of filtering during any epoch $j$ of \algMainAlmostFull, with probability $1 - \delta$, we have
$$|X_{i+1}| < \frac{1}{1 + \epsilon/3} |X_i|.$$
where $X_i$ and $X_{i+1}$ are the set $X$ before and after this $i$th  iteration and with sample complexity
$m =  \mathcal{O}\left(n\left( \frac{k +r}{\epsilon}\right)^2 \log\left(\frac{n}{\delta}\right)\right)$ at each round.
\end{lemma}
\begin{proof}
At a high level, since the surviving elements must have high value and a random set has low value, we can then use the thresholds to bound how many such surviving elements there can be while also having a random set of low value. To do so, we focus on the value of $f(R_i \cap X_{i+1})$ of the surviving elements $X_{i+1}$ in a random set $R_i \sim \D_{X_i}$.
\begin{align*}
 & \E \left[f_S(R_i \cap X_{i+1})\right] & \\
 \geq  & \E\left[\sum_{a \in R_i \cap X_{i+1}} f_{S \cup (R_i \cap X_{i+1} \setminus a)}(a)\right] & \text{submodularity}\\
 \geq  & \E\left[\sum_{a \in X_{i+1}} \mathds{1}_{a \in R_i} \cdot  f_{S \cup (R_i \setminus a)}(a)\right] & \text{submodularity}\\
 =  & \sum_{a \in X_{i+1}} \E\left[\mathds{1}_{a \in R_i} \cdot  f_{S \cup (R_i \setminus a)}(a)\right]. & \\
  = & \sum_{a \in X_{i+1}} \Pr\left[a \in R_i\right] \cdot \E\left[ f_{S \cup (R_i \setminus a)}(a) | a \in R_i\right] & \\
  \geq & \sum_{a \in X_{i+1}} \Pr\left[a \in R_i\right] \cdot \E\left[ f_{S \cup (R_i \setminus a)}(a) \right] & \text{submodularity}\\
 \geq & \sum_{a \in X_{i+1}} \Pr\left[a \in R\right] \cdot \left(v_S(X,a) -  \frac{\epsilon}{20k}\left(1 + \epsilon/2 \right)\left(1 -  \epsilon \right) \left(v^{\star} - f(S) \right)\right)  & \text{Lemma~\ref{lem:concentration} } \\
 \geq  & \sum_{a \in X_{i+1}} \Pr\left[a \in R_i\right] \cdot \left(\left(1 - \epsilon/20\right)\frac{1}{k}\left(1 + \epsilon/2 \right)\left(1 -  \epsilon \right) \left(v^{\star}  - f(S) \right)  \right) & \text{algorithm} \\
  =   & |X_{i+1}| \cdot \frac{k}{r|X_i|} \cdot \left(\left(1 - \epsilon/20\right) \frac{1}{k} \left(1 + \epsilon/2 \right)\left(1 - \epsilon \right) \left(v^{\star} - f(S) \right)\right) &  \\
  =  &  |X_{i+1}| \cdot \frac{1}{r|X_i|} \left(1 - \epsilon/20 \right)\cdot\left(1 + \epsilon/2 \right)\left(1 -  \epsilon \right) \left(v^{\star} - f(S) \right). & 
\end{align*}
Next, by the algorithm and by Lemma~\ref{lem:concentration},
$$  (1 + \epsilon/20) \frac{1}{r} \left(1 -  \epsilon \right) \left( v^{\star} - f(S) \right) \geq v_S(X) + \frac{\epsilon}{20}\frac{1}{r} \left(1 -  \epsilon \right) \left( v^{\star} - f(S) \right) > \E \left[f_S(R_i)\right].$$
By monotonicity, we get $\E \left[f_S(R_i)\right] \geq \E \left[f_S(R_i \cap X_{i+1})\right]$. Finally, by combining the above  inequalities, we  conclude that $$|X_{i+1}| \leq \frac{1}{ (1- \epsilon/20)^2(1+\epsilon/2)}|X_i| \leq \frac{1}{1 + \epsilon/3}|X_i| $$
where, with probability $1- \delta$, all the estimates hold with sample complexity $$m =  \mathcal{O}\left(n\left(\frac{k +r}{\epsilon}\right)^2 \log\left(\frac{n}{\delta}\right)\right)$$ per round by Lemma \ref{lem:concentration} and since $v^{\star} \geq (1 - \epsilon/20)\OPT$. 
\end{proof}

 \begin{lemma} 
\label{lem:optContrib2Full} Assume $(1 - \epsilon/20) \OPT \leq v^{\star} \leq  \OPT$ and $0 < \epsilon < 1/2
$. For any epoch $j$,
let $R_i \sim \U(X)$ be the random set at iteration $i$ of filtering during epoch $j$. For all $r, \rho > 0$, if epoch $j$ of \algMainAlmostFull \ has not ended after $\rho$ iterations of filtering, then, with probability $1 - \delta$,
$$\EU[R_1, \ldots, R_\rho]\left[f_{S_j^+ \cup \left(\cup_{i=1}^\rho R_i \right)}\left(O\right)\right] \geq \left(1 -   \frac{\rho}{r} - \epsilon/10\right) \cdot \left(\OPT - f(S_j)\right)$$
where $S_j$ and $S_j^+$ are the set $S$ at start and end of epoch $j$, with sample complexity $$m =  \mathcal{O}\left( \rho \left(\frac{\rho}{\epsilon}\right)^2\log\left(\frac{\rho}{\delta}\right)\right)$$ per epoch.
\end{lemma}
\begin{proof} By the condition to have a filtering iteration, a random set $R \sim \D$ must have low value at each of the \filter \ iterations:
\begin{align*}
& \EU[R_1, \ldots, R_\rho]\left[f_{S_j^+ \cup \left(\cup_{i=1}^\rho R_i \right)}\left(O\right)\right] \\
 =& \EU[R_1, \ldots, R_\rho]\left[f_{S_j^+ }\left(O\cup \left(\cup_{i=1}^\rho R_i \right)\right)\right] -  \EU[R_1, \ldots, R_\rho]\left[f_{S_j^+ }\left( \left(\cup_{i=1}^\rho R_i \right)\right)\right] & \\
\geq &\OPT - f(S_j^+) -  \EU[R_1, \ldots, R_\rho]\left[f_{S_j^+ }\left( \left(\cup_{i=1}^\rho R_i \right)\right)\right] & \text{monotonicity} \\
\geq &\OPT - f(S_j)  - \frac{\epsilon}{20} \left(\OPT - f(S)\right)  - \EU[R_1, \ldots, R_\rho]\left[f_{S_j^+ }\left( \left(\cup_{i=1}^\rho R_i \right)\right)\right] & \text{same epoch} \\
\geq &(1 - \epsilon/20)\left(\OPT - f(S)\right) - \sum_{i=1}^\rho \EU[R_i]\left[f_{S_j^+ }\left(  R_i \right)\right] & \text{subadditivity} \\
\geq &(1 - \epsilon/20)\left(\OPT - f(S)\right) - \sum_{i=1}^\rho \EU[R_i]\left[f_{S_{j,i}}\left(  R_i \right)\right] & \text{submodularity} \\ \geq &(1 - \epsilon/20)\left(\OPT - f(S)\right) - \sum_{i=1}^\rho
\EU[R_i]\Big[v_{S_{i,j}}(X_i) + \frac{\epsilon}{20\rho}\left(v^{\star} - f(S_j)\right)\Big] & \text{Lemma~\ref{lem:concentration}} \\
\geq &(1 - \epsilon/20)\left(\OPT - f(S)\right) - \sum_{i=1}^\rho
\EU[R_i]\Big[v_{S_{i,j}}(X_i) + \frac{\epsilon}{20\rho}\left(\OPT - f(S_j)\right)\Big] & v^{\star} \geq \OPT \\
\geq &(1 - \epsilon/20)\left(\OPT - f(S)\right)- \sum_{i=1}^\rho \left(\frac{1- \epsilon}{r}  + \frac{\epsilon}{20\rho}\right)\left(\OPT - f(S_{j})\right)] & \text{algorithm} \\
\geq &\left(1 - \epsilon/20  - \frac{\rho}{r} - \epsilon/20\right) \cdot \left(\OPT - f(S_j)\right). &
\end{align*}
where, with probability $1- \delta$,  the estimates hold for all $\rho$ iterations with sample complexity $m =  \mathcal{O}\left(\rho \left(\frac{\rho}{\epsilon}\right)^2 \log\left(\frac{\rho}{\delta}\right)\right)$ per round by Lemma \ref{lem:concentration} and since $v^{\star} \geq (1- \epsilon/20)\OPT$. 
\end{proof}

\begin{lemma}
\label{lem:TstarFull} Assume $(1 - \epsilon/20) \OPT \leq v^{\star} \leq  \OPT$ and $0 < \epsilon < 1/2$.
If $r \geq 20 \rho \epsilon^{-1}$,  then, with probability $1 - \delta$, the set  $X_\rho$ of elements that survive $\rho$ iterations of filtering at any epoch $j$ of \algMainAlmostFull \ satisfies 
$$f_{S_j^+}(X_{\rho}) \geq (\epsilon/4) \left(1 -  \epsilon \right) \left( v^{\star} - f(S_j) \right).$$
where $S_j$ and $S_j^+$ are the set $S$ at the start and end of epoch $j$ and with sample complexity
$m =  \mathcal{O}\left(\rho \left(\frac{k}{\epsilon}\right)^2 \log\left(\frac{\rho}{\delta}\right)\right)$ per epoch.
\end{lemma}
\begin{proof} Let $j$ be any epoch.
Similarly as for Lemma~\ref{lem:Tstar}, the proof  defines a subset $T$ of the optimal solution $O$ and then shows show that elements in $T$ survive $\rho$ iterations of filtering at epoch $j$ and show that $f_{S_j^+}(T) \geq (\epsilon/4) \left(1 - \epsilon \right) \left( \OPT - f(S_j) \right).$  We define the following marginal contribution $\Delta_\ell$ of each optimal element $o_\ell$:
$$\Delta_\ell := \EU[R_1, \ldots, R_\rho] \left[ f_{S_j^+ \cup O_{\ell-1}\cup \left(\cup_{i=1}^\rho R_i \setminus \{o_\ell\}\right) }(o_\ell)\right].$$

We define $T$ to be the set of optimal elements $o_\ell$ such that $\Delta_\ell \geq (1 - \epsilon/4) \Delta$ where
$$\Delta := \frac{1}{k}\left(1   - \frac{\rho}{r} - \epsilon/10\right) \cdot \left(\OPT - f(S_j)\right).$$
 We first argue that elements in $T$ survive $\rho$ iterations of filtering at epoch $j$. For element $o_\ell \in T$, we have
\begin{align*}
\Delta_\ell  \geq  (1 - \epsilon/4)\Delta & \geq \frac{1}{k}(1 - \epsilon/4)\left(1  - \frac{\rho}{r} - \epsilon/10\right)\cdot \left(\OPT - f(S_j)\right)  \\
& \geq \frac{1}{k}(1 - 5\epsilon/12)\cdot \left(\OPT - f(S_j)\right) \\
& \geq \frac{1}{k}(1+\epsilon/2) (1 - \epsilon) (1+\epsilon/20) \cdot \left(\OPT - f(S_j)\right).
\end{align*}
where the third inequality is by the condition on $r$.
Thus, at iteration $i \leq \rho$, by Lemma~\ref{lem:concentration} and by submodularity, 
\begin{align*}
v_{S_{j,i}}(X_i)  + \frac{\epsilon}{20} \cdot \frac{1}{k}(1+\epsilon/2) (1 - \epsilon)   \left(\OPT - f(S_j)\right)  & \geq 
\EU[R_i  ] \left[f_{S_{j,i} \cup (R_i \setminus \{o_\ell\})}(o_\ell) \right] \\
&  \geq  \EU[R_1, \ldots, R_\rho] \left[ f_{S_j^+ \cup O_{\ell-1}\cup \left(\cup_{i=1}^\rho R_i  \setminus \{o_\ell\}\right)}(o_\ell)\right] \\
&  = \Delta_\ell \\
& \geq \frac{1}{k}(1+\epsilon/2) (1 - \epsilon)  (1+\frac{\epsilon}{20})\cdot \left(\OPT - f(S_j)\right),
\end{align*}
where, with probability $1- \delta$, the estimates hold for all $\rho$ iterations with sample complexity
$m =  \mathcal{O}\left(\rho \left(\frac{k}{\epsilon}\right)^2 \log\left(\frac{\rho}{\delta}\right)\right)$. 
 Thus, $$v_{S_{j,i}}(X_i)  \geq \frac{1}{k}(1+\epsilon/2) (1 - \epsilon)  \cdot \left(\OPT - f(S_j)\right) \geq \frac{1}{k}(1+\epsilon/2) (1 - \epsilon)  \cdot \left(v^{\star} - f(S_j)\right)$$
and $o_\ell$ survives all iterations $i \leq \rho$, for all $o_\ell \in T$. 

Next, we argue that $f_{S_j^+}(T) \geq (\epsilon/4) \left(1 -  \epsilon \right) \left( \OPT - f(S_j) \right).$ Note that
$$\sum_{\ell =1}^k \Delta_\ell \geq   \EU[R_1, \ldots, R_\rho]\left[f_{S_j^+ \cup \left(\cup_{i=1}^\rho R_i \right)}\left(O\right)\right] \geq\left(1   - \frac{\rho}{r}- \epsilon/10\right) \cdot \left(\OPT - f(S_j)\right) =  k \Delta.$$

where the second inequality is by Lemma~\ref{lem:optContrib2Full}. Next, observe that
$$ \sum_{\ell =1}^k \Delta_\ell = \sum_{o_\ell \in T} \Delta_\ell + \sum_{j  \in O \setminus T} \Delta_\ell \leq \sum_{o_\ell \in T} \Delta_\ell + k (1 - \epsilon/4)\Delta.$$
By combining the two inequalities above, we get  $\sum_{o_\ell \in T} \Delta_\ell  \geq k \epsilon\Delta /4 $. Thus, by submodularity,
\begin{align*}
f_{S_j^+}(T)  \geq \sum_{o_\ell \in T} f_{S_j^+ \cup O_{\ell-1}}\left(o_\ell\right) 
 \geq \sum_{o_\ell \in T}\EU[R_1, \ldots, R_\rho] \left[ f_{S_j^+ \cup O_{\ell-1}\cup \left(\cup_{i=1}^\rho R_i  \setminus \{o_\ell\}\right)}(o_\ell)\right]  
 = \sum_{o_\ell \in T} \Delta_\ell  
 \geq k \epsilon \Delta/ 4. 
\end{align*}
We conclude that 
\begin{align*}
f_{S_j^+}(X_{\rho}) \geq f_{S_j^+}(T)  \geq k \epsilon \Delta /4= \frac{\epsilon}{4} \left( 1 - \frac{\rho}{r}  - \epsilon/10\right) \cdot \left(\OPT - f(S_j)\right) 
& \geq 
\frac{\epsilon}{4} \cdot  (1 -\epsilon)  \cdot \left(\OPT - f(S_j)\right) \\
& \geq 
\frac{\epsilon}{4} \cdot  (1 - \epsilon)  \cdot \left(v^{\star} - f(S_j)\right).
\end{align*}
where the first inequality is by monotonicity and since $T \subseteq X_\rho$ is  a set of surviving elements.
\end{proof}

\begin{lemma}
\label{lem:roundsFull}
Assume $(1 - \epsilon/20) \OPT \leq v^{\star} \leq  \OPT$ and constant $0 < \epsilon < 1/2$. For any epoch $j$, with probability $1 - \delta$,  there are at most $\log_{1+\epsilon/3}(n)$ iterations of filtering when the number of iterations of \algMainAlmostFull \ with $r = 20 \log_{1+\epsilon/3}(n) / \epsilon$ and with sample complexity
$m =  \mathcal{O}\left(n\left( \frac{k +r}{\epsilon}\right)^2 \log\left(\frac{n}{\delta}\right)\right)$ at each round.
\end{lemma}
\begin{proof}
If an epoch $j$  has not yet terminated after $\log_{1+\epsilon/3}(n)$ iterations of filtering, then,
by Lemma~\ref{lem:pruningFull}, at most $k/r$ elements survived these $\log_{1+\epsilon/3}(n)$ iterations. By Lemma~\ref{lem:TstarFull}, with the set $X_\rho$ of elements that survive these $\log_{1+\epsilon/3}(n)$ iterations is such that $f_{S_j^+}(X_{\rho}) \geq (\epsilon/4) \cdot \left(1 -  \epsilon \right) \left( v^{\star} - f(S_j) \right)$. Since there are at most $k/r$ surviving elements, $R = X_\rho$ for $R \sim \U(X_\rho)$ and 
$$\E\left[f_{S_j^+}(R)\right] \geq f_{S_j^+}(X_{\rho}) \geq \frac{\epsilon}{4} \cdot \left(1 -  \epsilon \right) \left( v^{\star} - f(S_j) \right) \geq \frac{1}{r}\cdot \left(1 -  \epsilon \right) \left( v^{\star} - f(S_j^+) \right)$$
where the last inequality is by monotonicity since $S_j \subseteq S_j^+$. Thus the current call to the \filter \ subroutine terminates and $X_{\rho}$ is added to $S_j^+$ by the algorithm. Next,
$$f_{S_j}(S_j^+ \cup X_{\rho}) \geq f_{S_j}( X_{\rho}) \geq  f_{S_j^+}( X_{\rho}) \geq  \frac{\epsilon}{4} \cdot \left(1 -  \epsilon \right) \left( \OPT - f(S_j) \right) \geq \frac{\epsilon}{20}\left( v^{\star} - f(S_j) \right) $$
where the first inequality is by monotonicity and the second by submodularity. Thus,  epoch $j$ ends.
\end{proof}

We are now ready to prove the main result for \algMainAlmostFull.

\begin{lemma}
\label{lem:algalmostFull}
Assume $(1 - \epsilon/20) \OPT \leq v^{\star} \leq  \OPT$. The \algMainAlmostFull  \ algorithm is a  $\mathcal{O}\left(\epsilon^{-1}\log_{1+\epsilon/3}(n)\right)$-adaptive algorithm that obtains, with probability $1 - \delta$, a $1 - 1/e -  \epsilon$ approximation, with $r = 20 \log_{1+\epsilon/3}(n) / \epsilon$. Its sample complexity
$m =  \mathcal{O}\left(n\left( \frac{k +\log_{1+\epsilon/3}(n) }{\epsilon}\right)^2 \log\left(\frac{n}{\delta}\right)\right)$ at each round.
\end{lemma}
\begin{proof}
First, consider the case where the algorithm terminates after $r$ iterations of adding elements to $S$. Let $S_i$ denote the solution $S$ at the $i$th iteration of \algMainAlmostFull.
The algorithm increases the value of the solution $S$ by at least $\left(1 -  \epsilon \right) \left( v^{\star} - f(S) \right) / r$  at every iteration with $k/r$ new elements. Thus,
$$f(S_i) \geq f(S_{i-1}) + \frac{1- \epsilon}{r} \left(v^{\star} - f(S_{i-1})\right)$$
and we obtain $$f(S) \geq \left(1 - e^{-(1- \epsilon)}\right)v^{\star} \geq \left(1 - e^{-(1- \epsilon)} - \epsilon/20\right)\OPT \geq \left(1 - \frac{1 + 2 \epsilon}{e} - \epsilon/20\right)\OPT \geq \left(1 - \frac{1}{e} - \epsilon\right)\OPT$$ similarly as for Theorem~\ref{thm:algone}. 

Next, consider the case where the algorithm terminated after $(\epsilon/20)^{-1}$ epochs. At every epoch $j$, the algorithm increases the value of the solution $S$ by $(\epsilon/20)(v^{\star} - f(S_j))$. Thus,
$$f(S_j) \geq f(S_{j-1}) + \frac{\epsilon}{20}\cdot\left(v^{\star} - f(S_{j-1})\right).$$
Similarly as in the first case, we get that after $(\epsilon/20)^{-1}$ epochs, $f(S) \geq (1 - e^{-1}) v^{\star} \geq (1 - e^{-1} - \epsilon/20) v^{\star}$.

The number of rounds is at most $ 20 \log_{1+\epsilon/3}(n) / \epsilon + (\epsilon/20)^{-1} \log_{1+\epsilon/3}(n)$ since there are at most $r = 20 \log_{1+\epsilon/3}(n) / \epsilon$ iterations of adding elements and at most $(\epsilon/20)^{-1}$ epochs, each of which with at most $\log_{1+\epsilon/3}(n)$  filtering iterations by Lemma~\ref{lem:roundsFull}.
\end{proof}

\thmalgFull*
\begin{proof} 
With $\log_{(1 - \epsilon/20)^{-1}}(n) $ different guesses $v^{\star}$ of $\OPT$, there is at least one $v^{\star}$ in \algMainFull \ that is such that 
$(1 - \epsilon/20) \OPT \leq v^{\star} \leq \OPT$. The solution to \algMainAlmostFull$(v^{\star})$  with such a $v^{\star}$ is then, with probability $1 - \delta$, a $1 -1/e - \epsilon$ approximation  
with sample complexity
$m =  \mathcal{O}\left(n\left( \frac{k +\log_{1+\epsilon/3}(n)}{\epsilon}\right)^2 \log\left(\frac{n}{\delta}\right)\right)$ at each round and with adaptivity $\mathcal{O}\left(\epsilon^{-1}\log_{1+\epsilon/3}(n)\right)$ by Lemma~\ref{lem:algalmostFull}. Since \algMainFull \ picks the best solution returned by all instances of \algMainAlmostFull, it also obtains with probability $1 - \delta$, a $1 -1/e - \epsilon$ approximation.  

Finally, since there are     $\log_{(1 - \epsilon/20)^{-1}}(n) $ non-adaptive  instances of \algMainAlmostFull, each with adaptivity $\mathcal{O}\left(\epsilon^{-1}\log_{1+\epsilon/3}(n)\right)$,  the total  number of adaptive rounds of  \algMainFull \ is $\mathcal{O}\left(\epsilon^{-1}\log_{1+\epsilon/3}(n)\right)$. The total query complexity per round over all guesses is $$m =  \mathcal{O}\left(n\log_{(1 - \epsilon/20)^{-1}}(n)\left( \frac{k +\log_{1+\epsilon/3}(n)}{\epsilon}\right)^2 \log\left(\frac{n}{\delta}\right)\right).$$ 
\end{proof}

\end{document}